\documentclass{amsart}
\usepackage{amsmath,amsthm,amsfonts, amssymb,amscd}
\usepackage[all]{xy}

\newtheorem{theorem}{Theorem}[section]
\newtheorem{lemma}[theorem]{Lemma}
\newtheorem{example}[theorem]{Example}
\newtheorem{proposition}[theorem]{Proposition}

\newcommand{\minusre}{\hspace{0.3em}\raisebox{0.3ex}{\sl \tiny /}\hspace{0.3em}}
\newcommand{\minusli}{\hspace{0.3em}\raisebox{0.3ex}{\sl \tiny $\setminus $}\hspace{0.3em}}
\newcommand{\lex}{\,\overrightarrow{\times}\,}

\newcommand{\Ker}{\mbox{\rm Ker}}

\newcommand{\Infinit}{\mbox{\rm Infinit}}

\def\ld{\mathord{\backslash}}
\def\rd{\mathord{/}}
\newcommand{\RDP}{\mbox{\rm RDP}}
\newcommand{\RIP}{\mbox{\rm RIP}}
\begin{document}
\title[Kite Pseudo Effect Algebras]{Kite Pseudo Effect Algebras}
\author[Anatolij Dvure\v{c}enskij]{Anatolij Dvure\v censkij$^{1,2}$}
\date{}%
\maketitle
\begin{center}  \footnote{


The paper has been supported by  Slovak Research and Development Agency under the contract APVV-0178-11, the grant VEGA No. 2/0059/12 SAV, and by
CZ.1.07/2.3.00/20.0051.
 }
Mathematical Institute,  Slovak Academy of Sciences,\\
\v Stef\'anikova 49, SK-814 73 Bratislava, Slovakia\\
$^2$ Depart. Algebra  Geom.,  Palack\'{y} University\\
17. listopadu 12, CZ-771 46 Olomouc, Czech Republic\\

E-mail: {\tt dvurecen@mat.savba.sk}
\end{center}

\begin{abstract}
We define a new class of pseudo effect algebras, called kite pseudo effect algebras, which is connected with partially ordered groups not necessarily with strong unit. In such a case, starting even with an Abelian po-group, we can obtain a noncommutative pseudo effect algebra. We show how such kite pseudo effect algebras are tied with different types of the Riesz Decomposition Properties. Kites are so-called perfect pseudo effect algebras, and we define conditions when kite pseudo effect algebras have the least non-trivial normal ideal.
\end{abstract}

\section{Introduction}

Effect algebras were introduced in \cite{FoBe} in order to describe events appearing in quantum mechanical measurements. They are partial algebras with a primary notion $+$,  which means that $a+b$ denotes the disjunction of mutually excluding events $a$ and $b.$ The main idea was to describe algebraically an appropriate model for the so-called POV-measures (positive operator valued measures) in the effect algebra $\mathcal E(H)$ of Hermitian operators between the zero and the identity operators of a real, complex or quaternionic Hilbert space $H.$ In the last two decades, they became an important class of so-called quantum structures which generalize Boolean algebras, orthomodular poset and orthomodular lattices, and orthoalgebras.

An important class of effect algebras cames from Abelian partially ordered groups (= po-groups) with strong unit $u$  as an interval $[0,u]$ in the positive cone. This is possible, for example,  whenever an effect algebra satisfies the Riesz Decomposition Property (RDP for short), see \cite{Rav}, which roughly speaking means possibility of a joint refinement of any two decompositions. Also every MV-algebra can serve as an example of effect algebras, \cite{Cha}, which can be characterized as a lattice ordered effect algebra with RDP. For more information about effect algebras, we recommend to see \cite{DvPu}. We note that MV-algebras describe many-valued reasoning. In addition, every lattice ordered effect algebra can be covered by a system of MV-algebras, \cite{Rie}.

Later the assumption that addition $+$ is commutative was canceled in \cite{DvVe1, DvVe2}, and pseudo effect algebras were introduced as a non-commutative generalization of effect algebras. This class also contains all pseudo MV-algebras which are a non-commutative generalization of MV-algebras, \cite{GeIo, Rac}. If a pseudo effect algebra satisfies a stronger type the Riesz Decomposition Property, RDP$_1$, then it is again an interval in a po-group (not necessarily Abelian) with strong unit, \cite{DvVe1, DvVe2}. We note that RDP$_1$ for effect algebras coincides with RDP, but for pseudo effect algebras they can be different. On the other side, every pseudo MV-algebra is an interval in a lattice ordered group (= $\ell$-group) with strong unit, \cite{Dvu1}, and it can be characterized as a pseudo effect algebra with RDP$_2,$ \cite{DvVe2}.

Recently,
in \cite{JiMo}, there was presented an interesting example of a pseudo BL-algebra which was described by the group of integers, $\mathbb Z,$ such that the universe of it was ordinal sum of its positive cone $\mathbb Z^+$ down and with $\mathbb Z^-\times \mathbb Z^-$ up and with a special kind of multiplication which resembles the wreath product. This idea was generalized in \cite{DvKo} for an any $\ell$-group and the whole theory of such algebras was presented there.  Since the shape of the pseudo effect algebra from \cite{JiMo} looks like a kite, we call them kites. Such algebras have their universe of the form $(G^+)^J$ down and $(G^-)^I$ up, where $G$ is an $\ell$-group, $I$ and $J$ are sets such that $|J|\le |I|,$ and multiplication uses two injections $\lambda,\rho: J\to I.$  If $|I|= |J|$ and $\lambda$ and $\rho$ are bijections,  the resulting algebra is a pseudo MV-algebra, \cite[Lem 3.4 (2)]{DvKo}.

In the present paper, we generalize these kites also for po-groups in order to obtain a new class of pseudo effect algebras, which we will call kite pseudo effect algebras. Such algebras are always perfect, i.e., every its element is either infinitesimal or co-infinitesimal (i.e. a negation of an infinitesimal). We note that an element $a$ of a pseudo effect algebra $E$ is infinitesimal if $na=a+\cdots +a$ ($n$-times addition of $a$)  is defined in $E$ for every integer $n\ge 1.$ We note that some families of perfect pseudo effect algebras were studied in \cite{Dvu4, DvXi, DvKr}.

The aim of the paper is to present the class of kite pseudo effect algebras. 
This class of pseudo effect algebras enriches the known store of examples of pseudo effect algebras, and it again shows  importance of po-groups for theory of pseudo effect algebras.

The paper is organized as follows. Section 2 gathers elements of pseudo effect algebras and pseudo MV-algebras. Kite pseudo effect algebras are described in Section 3 with some examples. In Section 4, there are kite pseudo effect algebras studied from the point of view of different types of the Riesz Decomposition Properties. Section 5 shows kite pseudo effect algebras as perfect pseudo effect algebras, and we present some representation of perfect pseudo effect algebras. Since our algebras are partial, instead of subdirect irreducibility we study in Section 6 the least non-trivial normal ideals of them and their representation.

\section{Elements of Pseudo Effect Algebras and Pseudo MV-algebras}

\subsection{Pseudo Effect Algebras}
By \cite{DvVe1, DvVe2}, we say that a {\it pseudo effect algebra} is  a partial algebra  $ E=(E; +, 0, 1)$, where $+$ is a partial binary operation and $0$ and $1$ are constants, such that for all $a, b, c
\in E$, the following holds

\begin{enumerate}
\item[(i)] $a+b$ and $(a+b)+c$ exist if and only if $b+c$ and
$a+(b+c)$ exist, and in this case $(a+b)+c = a+(b+c)$;

\item[(ii)]
  there is exactly one $d \in E$ and
exactly one $e \in E$ such that $a+d = e+a = 1$;

\item[(iii)]
 if $a+b$ exists, there are elements $d, e
\in E$ such that $a+b = d+a = b+e$;

\item[(iv)] if $1+a$ or $a+1$ exists, then $a = 0$.
\end{enumerate}

If we define $a \le b$ if and only if there exists an element $c\in
E$ such that $a+c =b,$ then $\le$ is a partial ordering on $E$ such
that $0 \le a \le 1$ for any $a \in E.$ It is possible to show that
$a \le b$ if and only if $b = a+c = d+a$ for some $c,d \in E$. We
write $c = a \minusre b$ and $d = b \minusli a.$ Then

$$ (b \minusli a) + a = a + (a \minusre b) = b,
$$
and we write $a^- = 1 \minusli a$ and $a^\sim = a\minusre 1$ for any
$a \in E.$ Then $a^-+a=1=a+a^\sim$ and $a^{-\sim}=a=a^{\sim-}$ for any $a\in E.$

For basic properties of pseudo effect algebras see \cite{DvVe1} and
\cite{DvVe2}. We note that a pseudo effect algebra is an {\it effect algebra}  iff $+$ is commutative.

A mapping $h$ from a pseudo effect algebra $E$ into another one $F$ is said to be a {\it homomorphism} if (i) $h(1)=1,$ and (ii) if $a+b$ is defined in $E,$ so is defined $h(a)+h(b)$ and $h(a+b)= h(a)+h(b).$ A homomorphism $h$ is an {\it isomorphism} if $h$ is injective a surjective and also $h^{-1}$ is a homomorphism.

We remind that a {\it po-group} (= partially ordered group) is a
group $G=(G;+,-,0)$ endowed with a partial order $\le$ such that if $a\le b,$ $a,b
\in G,$ then $x+a+y \le x+b+y$ for all $x,y \in G.$  We denote by
$G^+:=\{g \in G: g \ge 0\}$ and $G^-:=\{g \in G: g \le 0\}$ the {\it positive cone} and the {\it negative cone} of $G.$ If, in addition, $G$
is a lattice under $\le$, we call it an $\ell$-group (= lattice
ordered group). An element $u \in G^+$ is said to be a {\it strong unit} (or an {\it order unit}) if, given $g \in G,$ there is an integer $n \ge 1$ such that $g \le nu.$ The pair $(G,u),$ where $u$ is a fixed strong unit of $G,$ is said to be a {\it unital po-group}.

We denote by  $\mathbb Z$ the commutative $\ell$-group of integers.

A po-group $G$ is said to be {\it directed} if, given $g_1,g_2 \in G,$ there is an element $g \in G$ such that $g \ge g_1,g_2.$ Equivalently, $G$ is directed iff every element $g\in G$ can be expressed as a difference of two positive elements of $G.$ For example, every $\ell$-group or every  po-group with strong unit is directed.
For more information on po-groups and $\ell$-groups we recommend the books \cite{Dar, Fuc, Gla}.

Now let  $G$ be a po-group and fix $u \in G^+.$ If we set $\Gamma(G,u):=[0,u]=\{g \in G: 0 \le g \le u\},$ then $\Gamma(G,u)=(\Gamma(G,u); +,0,u)$ is a pseudo effect algebra, where $+$ is the restriction of the group addition $+$ to $[0,u],$ i.e. $a+b$ is defined in $\Gamma(G,u)$ for $a,b \in \Gamma(G,u)$ iff $a+b \in \Gamma(G,u).$ Then $a^-=u-a$ and $a^\sim=-a+u$ for any $a \in \Gamma(G,u).$ A pseudo effect algebra which is isomorphic to some $\Gamma(G,u)$ for some po-group $G$ with $u>0$ is said to be an {\it interval pseudo effect algebra}.

A pseudo effect algebra $E$ is said to be {\it  symmetric} if $a^-=a^\sim$ for each $a\in E.$ We note that if $E$ is a symmetric pseudo effect algebra, then it is not necessarily an effect algebra. Indeed, if $G$ is a non-commutative po-group, let $\mathbb Z\lex G$ denote the lexicographic product of the group of integers, $\mathbb Z$, with $G.$ Then $u=(1,0)$ is a strong unit, and $E=\Gamma(\mathbb Z \lex G,(1,0))$ is a pseudo effect algebra with $a^-=a^\sim$ for all $a\in E,$ but $E$ is not an effect algebra. We note that if $(G,u)$ is a unital po-group with RDP (for the definition of RDP see Subsection 2.3),  the pseudo effect algebra $\Gamma(G,u)$ is symmetric iff $u$ belongs to the commutative center $C(G):=\{c\in G: c+g=g+c, \forall \ g \in G\}.$

\subsection{Pseudo MV-algebras and MV-algebras}

We remind that an MV-algebra is an algebra $M = (M;\oplus, ^*,0,1)$
of type (2,1,0,0) such that, for all $a,b,c \in M$, we have

\begin{enumerate}
\item[(i)]  $a \oplus  b = b \oplus a$;
\item[(ii)] $(a\oplus b)\oplus c = a \oplus (b \oplus c)$;
\item[(iii)] $a\oplus 0 = a;$
\item[(iv)] $a\oplus 1= 1;$
\item[(v)] $(a^*)^* = a;$
\item[(vi)] $0^* = 1;$
\item[(vii)] $(a^*\oplus b)^*\oplus b=(a\oplus b^*)^*\oplus a.$
\end{enumerate}

If we define $a\le b$ iff there is an element $c\in M$ such that $a\oplus c=b$, then $\le$ is a partial order and $a\vee b = (a^*\oplus b)^*\oplus b.$  With respect to this order, $M$ is a distributive lattice. In addition, we can define also another  binary operation $a\odot b= (a^*\oplus b^*)^*.$\footnote{Notational convention: $\odot$ binds stronger than $\oplus$.}

According to \cite{GeIo}, a {\it pseudo MV-algebra} is an algebra $ M=(M;
\oplus,^-,^\sim,0,1)$ of type $(2,1,1,$ $0,0)$ such that the
following axioms hold for all $x,y,z \in M$ with an additional
binary operation $\odot$ defined via $$ y \odot x =(x^- \oplus y^-)
^\sim $$
\begin{enumerate}
\item[{\rm (A1)}]  $x \oplus (y \oplus z) = (x \oplus y) \oplus z;$

\item[{\rm (A2)}] $x\oplus 0 = 0 \oplus x = x;$

\item[{\rm (A3)}] $x \oplus 1 = 1 \oplus x = 1;$

\item[{\rm (A4)}] $1^\sim = 0;$ $1^- = 0;$

\item[{\rm (A5)}] $(x^- \oplus y^-)^\sim = (x^\sim \oplus y^\sim)^-;$

\item[{\rm (A6)}] $x \oplus (x^\sim \odot y) = y \oplus (y^\sim
\odot x) = (x \odot y^-) \oplus y = (y \odot x^-) \oplus
x;$

\item[{\rm (A7)}] $x \odot (x^- \oplus y) = (x \oplus y^\sim)
\odot y;$

\item[{\rm (A8)}] $(x^-)^\sim= x.$
\end{enumerate}

A pseudo MV-algebra $M$ is an MV-algebra iff $a\oplus b = b\oplus a,$ $a,b \in M.$

For example, if $u$ is a strong unit of a (not necessarily Abelian)
$\ell$-group $G$,
$$\Gamma(G,u) := [0,u]
$$
and
\begin{eqnarray*}
x \oplus y &:=&
(x+y) \wedge u,\\
x^- &:=& u - x,\\
x^\sim &:=& -x +u,\\
x\odot y&:= &(x-u+y)\vee 0,
\end{eqnarray*}
then $\Gamma(G,u)=(\Gamma(G,u);\oplus, ^-,^\sim,0,u)$ is a pseudo MV-algebra.

The basic result on theory of pseudo MV-algebras \cite{Dvu1} is the following representation theorem.

\begin{theorem}\label{th:2.1}
Every pseudo MV-algebra is an interval $\Gamma(G,u)$ in a unique (up to isomorphism)  unital $\ell$-group $(G,u).$

In addition, the functor $\Gamma: (G,u) \mapsto \Gamma(G,u)$ defines a categorical equivalence between the variety of pseudo MV-algebras and the category of unital $\ell$-groups.
\end{theorem}

\subsection{Riesz Decomposition Properties}
The important Riesz Decomposition Property, (RDP for short) means roughly speaking a property that every two decompositions have a joint refinement. For non-commutative structures we can define more types of RDPs and some of them in commutative structures coincide.

The following kinds of the Riesz Decomposition properties were introduced in \cite{DvVe1,DvVe2} and are crucial for the study of pseudo effect algebras.

We say that a  po-group $G$ satisfies

\begin{enumerate}
\item[(i)]
the {\it Riesz Interpolation Property} (RIP for short) if, for $a_1,a_2, b_1,b_2\in G,$  $a_1,a_2 \le b_1,b_2$  implies there exists an element $c\in G$ such that $a_1,a_2 \le c \le b_1,b_2;$

\item[(ii)]
\RDP$_0$  if, for $a,b,c \in G^+,$ $a \le b+c$, there exist $b_1,c_1 \in G^+,$ such that $b_1\le b,$ $c_1 \le c$ and $a = b_1 +c_1;$

\item[(iii)]
\RDP\  if, for all $a_1,a_2,b_1,b_2 \in G^+$ such that $a_1 + a_2 = b_1+b_2,$ there are four elements $c_{11},c_{12},c_{21},c_{22}\in G^+$ such that $a_1 = c_{11}+c_{12},$ $a_2= c_{21}+c_{22},$ $b_1= c_{11} + c_{21}$ and $b_2= c_{12}+c_{22};$

\item[(iv)]
\RDP$_1$  if, for all $a_1,a_2,b_1,b_2 \in G^+$ such that $a_1 + a_2 = b_1+b_2,$ there are four elements $c_{11},c_{12},c_{21},c_{22}\in G^+$ such that $a_1 = c_{11}+c_{12},$ $a_2= c_{21}+c_{22},$ $b_1= c_{11} + c_{21}$ and $b_2= c_{12}+c_{22}$, and $0\le x\le c_{12}$ and $0\le y \le c_{21}$ imply  $x+y=y+x;$

\item[(v)]
\RDP$_2$  if, for all $a_1,a_2,b_1,b_2 \in G^+$ such that $a_1 + a_2 = b_1+b_2,$ there are four elements $c_{11},c_{12},c_{21},c_{22}\in G^+$ such that $a_1 = c_{11}+c_{12},$ $a_2= c_{21}+c_{22},$ $b_1= c_{11} + c_{21}$ and $b_2= c_{12}+c_{22}$, and $c_{12}\wedge c_{21}=0.$

\end{enumerate}

If, for $a,b \in G^+,$ we have for all $0\le x \le a$ and $0\le y\le b,$ $x+y=y+x,$ we denote this property by $a\, \mbox{\rm \bf com}\, b.$

The RDP will be denoted by the following table:

$$
\begin{matrix}
a_1  &\vline & c_{11} & c_{12}\\
a_{2} &\vline & c_{21} & c_{22}\\
  \hline     &\vline      &b_{1} & b_{2}
\end{matrix}\ \ .
$$

For Abelian po-groups, RDP, RDP$_1,$ RDP$_0$ and RIP are equivalent.

By \cite[Prop 4.2]{DvVe1} for directed po-groups, we have
$$
\RDP_2 \quad \Rightarrow \RDP_1 \quad \Rightarrow \RDP \quad \Rightarrow \RDP_0 \quad \Leftrightarrow \quad  \RIP,
$$
but the converse implications do not hold, in general.  A directed po-group $G$ satisfies \RDP$_2$ iff $G$ is an $\ell$-group, \cite[Prop 4.2(ii)]{DvVe1}.

We say that a pseudo effect algebra $E$ satisfies the above types of the Riesz decomposition properties, if in the definition of RDP's, we change $G^+$ to $E.$

The basic result of pseudo effect algebras, which binds RDPs with pseudo effect algebras, is the following representation theorem \cite[Thm 7.2]{DvVe2}:

\begin{theorem}\label{th:2.2}
For every pseudo effect algebra with \RDP$_1,$ there is a unique (up to isomorphism of unital po-groups) unital po-group $(G,u)$ with \RDP$_1$\ such that $E \cong \Gamma(G,u).$

In addition, $\Gamma$ defines a categorical equivalence between the category of pseudo effect algebras with \RDP$_1$ and the category of unital po-groups with \RDP$_1.$
\end{theorem}

Let $M$ be a pseudo MV-algebra. We define a partial binary operation, $+,$ on $M$ such that $a+b$ is defined in $M$ iff $a\odot b = 0,$ and in such a case, $a+ b:= a \oplus b.$ Then $(M;+,0,1)$ is a pseudo effect algebra with RDP$_2.$ Conversely, every pseudo effect algebra with RDP$_2$ can transformed into a pseudo MV-algebra. Equivalently, a pseudo effect algebra $E$ with RDP$_1$ satisfies RDP$_2$ iff $E$ is a lattice, see \cite[Thm 8.8]{DvVe2}.

Thus the class of pseudo MV-algebras forms an important subclass of the class of pseudo effect algebras.

\section{Kite Pseudo Effect Algebras}

Take a po-group $G=(G;\cdot,^{-1},e)$ which is written multiplicatively with an inversion $^{-1}$ and with the identity element $e$ equipped with a partial order $\le.$

Let $I$ be a set. Since only the cardinality of $I$ is important for the following construction,  we assume that $I$ is an ordinal. Define an algebra whose
universe is the set $(G^+)^I \uplus (G^-)^I,$ where $\uplus$ denotes a union of disjoint sets.
We order its universe by keeping the original co-ordinatewise ordering
within $(G^+)^I$ and $(G^-)^I$, and setting $x\leq y$ for all
$x\in(G^+)^I$, $y\in(G^-)^I$. Then $\le$ is a partial order on $(G^+)^I \uplus (G^-)^I.$ Then the element $e^I:=\langle e: i \in I\rangle$ appears twice: at the bottom of $(G^+)^I$ and at the top of $(G^-)^I$. To
avoid confusion in the definitions below, we adopt a convention of writing
$a_i^{-1},b_i^{-1}, \dots$ for co-ordinates of elements of $(G^-)^I$ and
$f_j,g_j,\dots$ for co-ordinates of elements of $(G^+)^I$. In particular, we
will write $e^{-1}$ for $e$ as an element of $G^-$. We also put $1$ for the
constant sequence $(e^{-1})^I:=\langle e^{-1}: i \in
I\rangle$ and $0$ for the constant sequence $e^I$. Then $0$ and $1$ are the least and greatest elements of $(G^+)^I \uplus (G^-)^I.$

Let $\lambda, \rho: I \to I$ be bijections.

First of all we assume that a po-group $G$ is an $\ell$-group. In such a case, according to \cite{DvKo}, we define multiplication $\cdot$ on $(G^+)^I \uplus (G^-)^I$ by

\begin{align*}
\langle a_i^{-1}\colon i\in I\rangle\cdot\langle b_i^{-1}\colon i\in I\rangle &=
  \langle(b_ia_i)^{-1}\colon i\in I\rangle\\
\langle a_i^{-1}\colon i\in I\rangle\cdot\langle f_j\colon j\in I\rangle &=
  \langle a_{\lambda(j)}^{-1}f_j\vee e\colon j\in I\rangle\\
\langle f_j\colon j\in I\rangle\cdot\langle a_i^{-1}\colon i\in I\rangle &=
  \langle f_ja_{\rho(j)}^{-1}\vee e\colon j\in I\rangle\\
\langle f_j\colon j\in I\rangle\cdot\langle g_j\colon j\in I\rangle &=
  \langle e\colon j\in I\rangle = 0.\\
\end{align*}

Similarly, we define divisions, $\rd$ and $\ld$, as follows
\begin{align*}
\langle a_i^{-1}\colon i\in I\rangle\ld\langle b_i^{-1}\colon i\in I\rangle &=
 \langle a_ib_i^{-1}\wedge e^{-1}\colon i\in I\rangle\\
\langle b_i^{-1}\colon i\in I\rangle\rd\langle a_i^{-1}\colon i\in I\rangle &=
 \langle b_i^{-1}a_i\wedge e^{-1}\colon i\in I\rangle\\
\langle a_i^{-1}\colon i\in I\rangle\ld\langle f_j\colon j\in I\rangle &=
 \langle a_{\lambda(j)}f_j\colon j\in I\rangle\\
\langle f_j\colon j\in I\rangle\rd\langle a_i^{-1}\colon i\in I\rangle &=
 \langle f_ja_{\rho(j)}\colon j\in I\rangle\\
\langle f_j\colon j\in I\rangle\ld\langle g_j\colon j\in I\rangle &= \langle a_i^{-1}\colon i\in I\rangle,\\
\text{ where } a_i^{-1} &=\begin{cases}
f_{\rho^{-1}(i)}^{-1}g_{\rho^{-1}(i)}\wedge e^{-1} &
\text{ if } \rho^{-1}(i) \text{ is defined}\\
e^{-1} & \text{ otherwise}
\end{cases}\\
\langle g_j\colon j\in I\rangle\rd\langle f_j\colon j\in I\rangle &= \langle b_i^{-1}\colon i\in I\rangle,\\
\text{ where } b_i^{-1} &=\begin{cases}
g_{\lambda^{-1}(i)}f_{\lambda^{-1}(i)}^{-1}\wedge e^{-1} &
\text{ if } \lambda^{-1}(i) \text{ is defined}\\
e^{-1} & \text{ otherwise},
\end{cases}\\
\langle a_i^{-1}\colon i\in I\rangle\rd \langle f_j\colon j\in
I\rangle &=(e^{-1})^I= \langle f_j\colon j\in I\rangle \ld \langle
a_i^{-1}\colon i\in I\rangle.
\end{align*}

Let us define additional operations
$$
x\oplus y :=(x^-\cdot y^-)^\sim, \quad x,y \in (G^+)^I \uplus (G^-)^I, \eqno(  3.1)
$$
$\odot = \cdot,$ and $x^-=0\rd x$ and $x^\sim = x\ld 0$ as the {\it right} and {\it left negation}, respectively, of $x.$ According to \cite[Lem 3.3]{DvKo}, for the algebra $K_{I}^{\lambda,\rho}(G)_{mv}= ((G^+)^I \uplus (G^-)^I; \oplus, ^-,^\sim,0,1)$, we have:


\begin{theorem}\label{th:3.1}
Let $G$ be an $\ell$-group. The algebra $K_{I}^{\lambda,\rho}(G)_{mv}$ is a pseudo MV-algebra.
\end{theorem}

We note that we can define the standard total MV-operation $\oplus$ on the kite $K_{I}^{\lambda,\rho}(G)_{mv}$ as follows

Due to the Theorem \ref{th:3.1}, we will call  the algebra $K_{I}^{\lambda,\rho}(G)_{mv}$ a \emph{kite} or, more precisely, a {\it kite pseudo MV-algebra} of $G.$

For example, if $G=O$ is the zero group, then $K_{I}^{\lambda,\rho}(O)_{mv}$ is a two-element Boolean algebra. If $G=\mathbb Z$ and $|I|=1,$ then
$K_{I}^{\lambda,\rho}(\mathbb Z)_{mv} \cong \Gamma(\mathbb Z \lex \mathbb Z,(1,0))$ is the Chang MV-algebra. If $G$ is arbitrary and $|I|=1,$ then $K_{I}^{\lambda,\rho}(G)_{mv}$ is isomorphic to the pseudo MV-algebra $\Gamma(\mathbb Z \lex G, (1,0)).$

In the same way as for pseudo effect algebras, we say that a pseudo MV-algebra $M$ is {\it symmetric} if $x^-=x^\sim$ for $x \in M.$

It is easy to show that:

\begin{proposition}\label{pr:3.2}
For the negations in the kite $K_{I}^{\lambda,\rho}(G)_{mv},$ we have
\begin{align*}
\langle a_i^{-1}\colon i\in I\rangle^\sim &= \langle a_{\lambda(j)}\colon j\in I\rangle\\
\langle a_i^{-1}\colon i\in I\rangle^-&= \langle a_{\rho(j)}\colon j\in I\rangle\\
\langle f_j\colon j\in I\rangle^\sim &= \langle f^{-1}_{\rho^{-1}(i)}\colon i\in I\rangle\\
\langle f_j\colon j\in I\rangle^-&= \langle f^{-1}_{\lambda^{-1}(i)}\colon i\in I\rangle.
\end{align*}
\end{proposition}

\begin{proposition}\label{pr:3.3}
Let $G$ be a non-trivial $\ell$-group.

{\rm (1)} The kite $K_{I}^{\lambda,\rho}(G)_{mv}$ is a symmetric pseudo MV-algebra if and only if $\lambda =\rho.$

{\rm (2)} The kite $K_{I}^{\lambda,\rho}(G)_{mv}$ is an MV-algebra if and only if  $\lambda =\rho$ and $G$ is an Abelian $\ell$-group.
\end{proposition}

\begin{proof}

(1) By Proposition \ref{pr:3.2}, it is clear that if $\lambda =\rho,$ then the kite $K_{I}^{\lambda,\rho}(G)_{mv}$ is  symmetric.

Conversely, let  $K_{I}^{\lambda,\rho}(G)_{mv}$ be a symmetric pseudo MV-algebra. From Proposition \ref{pr:3.2}, we conclude $a_{\lambda(j)}=a_{\rho(j)}$ for any $j \in I.$ Since $G \ne O,$ $G$ has at least two different elements and combining different elements, we have $\lambda = \rho.$

(2) From the definition of multiplication $\cdot$ and divisions, $\ld, \rd,$ we easily conclude that if $G$ is an Abelian $\ell$-group and $\lambda =\rho,$ then $K_{I}^{\lambda,\rho}(G)_{mv}$ is an MV-algebra.

Conversely, assume that $K_{I}^{\lambda,\rho}(G)_{mv}$ is an MV-algebra.  Then it is symmetric and from (1), we have $\lambda(j)=\rho(j)$ for $j \in I.$ From the first line of the multiplication list, we conclude $G$ is Abelian.
\end{proof}

We note that if $\oplus$ is defined by (3.1), the partial operation $+$ induced by $\oplus$ is defined as follows:
$x +y$ is defined iff $x\cdot y=0,$ equivalently, $y \le x^-,$ or equivalently, $x \le y^\sim,$ and in such a case, $x+y := x\oplus y.$ Therefore, $(G^+)^I \uplus (G^-)^I$ endowed with this partial operation $+$ and $0,1$ is a pseudo effect algebra. Then the right and left negations of any element are defined according to Proposition \ref{pr:3.2}.

Motivated by this idea, we define a kite pseudo effect algebra starting with an arbitrary po-group, not necessarily lattice ordered.

\begin{theorem}\label{th:3.4}
Let $G$ be a po-group and $\lambda,\rho:I\to I$ be bijections. Let us endow the set $(G^+)^I \uplus (G^-)^I$ with $0,$ $1$ and with a partial operation $+$ as follows,

$$\langle a_i^{-1}\colon i\in I\rangle + \langle b_i^{-1}\colon i\in I\rangle \eqno(I)$$
is not defined;

$$ \langle a_i^{-1}\colon i\in I\rangle + \langle f_j\colon j\in I\rangle:= \langle a_i^{-1}f_{\rho^{-1}(i)}\colon i\in I\rangle \eqno(II)
$$
whenever  $f_{\rho^{-1}(i)}\le a_i,$ $i \in I;$

$$ \langle f_j\colon j\in I\rangle+ \langle a_i^{-1}\colon i\in I\rangle  := \langle f_{\lambda^{-1}(i)} a_i^{-1}\colon i\in I\rangle \eqno(III)
$$
whenever  $f_{\lambda^{-1}(i)}\le a_i,$ $i \in I,$

$$
\langle f_j\colon j\in I\rangle + \langle g_j\colon j\in I\rangle:= \langle f_j g_j\colon j\in I\rangle
\eqno(IV)
$$
for all $\langle f_j\colon j\in I\rangle$ and $\langle g_j\colon j\in I\rangle.$
Then the partial algebra $K^{\lambda,\rho}_I(G)_{ea}:=((G^+)^I \uplus (G^-)^I; +,0,1)$ becomes a pseudo effect algebra. The left and right negations are defined by Proposition {\rm \ref{pr:3.2}.}

The pseudo effect algebra $K^{\lambda,\rho}_I(G)_{ea}$ is symmetric if and only if $\lambda =\rho.$

If $G$ is directed, the pseudo effect algebra $K^{\lambda,\lambda}_I(G)_{ea}$ is an effect algebra if and only if $G$ is an Abelian po-group.

If $G$ is an $\ell$-group, then $K^{\lambda,\rho}_I(G)_{ea}$ is a pseudo effect algebra with \RDP$_2.$
\end{theorem}

\begin{proof}
(i) To prove  associativity of $+$ we have 8 cases, and it is enough to prove only the following case  because all other are simple.

\begin{align*}
(\langle f_j\colon j\in I\rangle + \langle b_i^{-1}\colon i\in I\rangle)&+
\langle h_j\colon j\in I\rangle
=\langle f_{\lambda^{-1}(i)}b_i^{-1}\colon i\in I\rangle + \langle h_j\colon j\in I\rangle\\
&=\langle f_{\lambda^{-1}(i)}b_i^{-1}h_{\rho^{-1}(i)}\colon i\in I\rangle\\
&=\langle f_j\colon j\in I\rangle + (\langle b_i^{-1}\colon i\in I\rangle+
\langle h_j\colon j\in I\rangle).
\end{align*}
If the left hand side of the first line exists, then $f_{\lambda^{-1}(i)}\le b_i$ and $h_{\rho^{-1}(i)}\le b_i f_{\lambda^{-1}(i)}^{-1}.$ The second inequality entails $h_{\rho^{-1}(i)}\le b_i f_{\lambda^{-1}(i)}^{-1} \le b_i$ for $i \in I$ so that $ \langle b_i^{-1}\colon i\in I\rangle+
\langle h_j\colon j\in I\rangle$ and $\langle f_j\colon j\in I\rangle + (\langle b_i^{-1}\colon i\in I\rangle+
\langle h_j\colon j\in I\rangle)$ are defined.

Conversely, let the later elements be defined, then $h_{\rho^{-1}(i)} \le b_i$ and $h_{\rho^{-1}(i)}\le b_i f_{\lambda^{-1}(i)}^{-1}.$ Then $f_{\lambda^{-1}(i)}\le b_i$ and $h_{\rho^{-1}(i)}\le b_i f_{\lambda^{-1}(i)}^{-1},$ so that the elements $\langle f_j\colon j\in I\rangle + \langle b_i^{-1}\colon i\in I\rangle$ and $(\langle f_j\colon j\in I\rangle + \langle b_i^{-1}\colon i\in I\rangle)+
\langle h_j\colon j\in I\rangle$ are defined.

(ii) If we define the left and right negations in the same way as it is in Proposition \ref{pr:3.2}, we have  $x^- + x=x+x^\sim = 1.$ The uniqueness of them can be proved thanks to properties of po-groups.

(iii) Assume $x+y$ is defined. For example, let $\langle a_i^{-1}\colon i\in I\rangle + \langle f_j\colon j\in I\rangle:= \langle a_i^{-1}f_{\rho^{-1}(i)}\colon i\in I\rangle$ be defined.
If we set $h_j:= a^{-1}_{\lambda(j)}f_{\rho^{-1}(\lambda(j))} a_{\lambda (j)}$ and $b_i^{-1}:= f_{\lambda^{-1}(i)}^{-1}a_i^{-1}f_{\rho^{-1}(i)},$ we have $h_j\ge e$ and $b_i^{-1}=f_{\lambda^{-1}(i)}^{-1}a_i^{-1}f_{\rho^{-1}(i)} \le f^{-1}_{\lambda^{-1}(i)}\le e.$ Therefore,
$$
\langle h_j\colon j\in I\rangle + \langle a_i^{-1}\colon i\in I\rangle= \langle a_i^{-1}\colon i\in I\rangle + \langle f_j\colon j\in I\rangle =
\langle f_j\colon j\in I\rangle + \langle b_i^{-1}\colon i\in I\rangle.
$$

(iv) Let $1 +x$ or $x+1$ be defined, then it easy to show that $x=0.$

Summarizing (i)--(iv), we have $K^{\lambda,\rho}_I(G)_{ea}$ is a pseudo effect algebra.

In the same way as in the proof of Proposition \ref{pr:3.3}(1), we can prove that $K^{\lambda,\rho}_I(G)_{ea}$ is a symmetric pseudo effect algebra iff $\lambda = \rho.$

If $G$ is an Abelian po-group, then $K^{\lambda,\lambda}_I(G)_{ea}$ is an effect algebra. Conversely, if $K^{\lambda,\lambda}_I(G)_{ea}$ is an effect algebra, then by $(IV),$ we have $G$ is Abelian.

Finally, let $G$ be an $\ell$-group.  We note that if $G$ is an $\ell$-group, then the kite pseudo effect algebra $K^{\lambda,\rho}_I(G)_{ea}$ is obtained also as a pseudo effect algebra from the pseudo MV-algebra $K^{\lambda,\rho}_I(G)_{mv},$ and every pseudo MV-algebra satisfies RDP$_2.$
\end{proof}

The pseudo effect algebra $K_{I}^{\lambda,\rho}(G)_{ea}= ((G^+)^I \uplus (G^-)^I; \cdot, +,0,1)$ is also said to be a {\it kite}, or more precisely, a {\it kite pseudo effect algebra} of $G.$

For example, if $G$ is a po-group and $|I|=1,$ then $K_{I}^{\lambda,\lambda}(G)$ is isomorphic to the pseudo effect algebra $\Gamma(\mathbb Z\lex G,(1,0))_{ea}.$ If $|I|=n,$ then $K_{I}^{id,id}(G)_{ea}$ is isomorphic to the pseudo effect algebra  $\Gamma(\mathbb Z \lex G^n,(1,0,\ldots,0)).$

If  $G=\mathbb Z$ and $|I|=n,$ the corresponding kite can be connected with the Scrimger $\ell$-groups as we can see in the next examples.

\begin{example}\label{ex:3.5}
Let $|I|=n$ for $n\ge 2,$ and $\lambda(i)=i,$ $\rho(i) = i-1\ (\mathrm{mod}\, n).$
We put $G_n = \mathbb Z \lex (\mathbb Z^n)$ which is ordered lexicographically. We define the addition $*$ on $G_n$ as follows


$$
(m_1,x_0,\ldots,x_{n-1})*(m_2,y_0,\ldots,y_{n-1})=(m_1+m_2,x_0+y_{0+m_1},\ldots,
x_{n-1}+y_{n-1+m_1}),
$$
where addition of the subscripts is added by $\mathrm{mod}\, n.$
Then $G_n$ with $*$ is an $\ell$-group, where the inversion is given by $-(m,a_0,\ldots,a_{n-1})=(-m,-a_{-m},\ldots,-a_{n-1-m}),$ and the element $u_n=(1,0,\ldots,0)$ is a strong unit.
Then $K_{I}^{\rho,\lambda}(\mathbb Z)_{ea}$ is isomorphic to the pseudo effect algebra $\Gamma(G_n,u_n)$ with \RDP$_2.$
\end{example}

\begin{example}\label{ex:3.6}
Let $|I|=n$ for $n\ge 2,$ and $\lambda(i)=i,$ $\rho(i) = i-1\ (\mathrm{mod}\, n).$
We put $G_n = \mathbb Z \lex (\mathbb Z^n)$ which is ordered lexicographically. We define the addition $*_1$ on $G_n$ as follows


$$
(m_1,x_0,\ldots,x_{n-1})*_1(m_2,y_0,\ldots,y_{n-1})=(m_1+m_2,x_{m_2}+y_{0},\ldots,
x_{n-1+m_2}+y_{n-1}),
$$
where addition of the subscripts is added by $\mathrm{mod}\, n.$
Then $G_n$ with $*_1$ is an $\ell$-group, called the Scrimger $\ell$-group, where the inversion is given by $-(m,a_0,\ldots,a_{n-1})=(-m,-a_{-m},\ldots,-a_{n-1-m}),$ and the element $u_n=(1,0,\ldots,0)$ is a strong unit.
Then $K_{I}^{\lambda,\rho}(\mathbb Z)_{ea}$ is isomorphic to the pseudo effect algebra $\Gamma(G_n,u_n)$ with \RDP$_2.$
\end{example}

\begin{example}\label{ex:3.7}
Let $|I|=n$ for $n\ge 2,$ and $\lambda(i)=i,$ $\rho(i) = i-1\ (\mathrm{mod}\, n),$ and let $G$ be a po-group with a neutral element $e.$
We put $G_n(G) = \mathbb Z \lex (G^n)$ which is ordered lexicographically. We define the addition $*$ on $G_n(G)$ as follows

$$
(m_1,x_0,\ldots,x_{n-1})*(m_2,y_0,\ldots,y_{n-1})=(m_1+m_2,x_0y_{0+m_1},\ldots,
x_{n-1}y_{n-1+m_1}),
$$
where addition of the subscripts is added by $\mathrm{mod}\, n.$
Then $G_n(G)$ is a po-group, where the inversion is given by $(m,a_0,\ldots,a_{n-1})^{-1}=(-m,a_{-m}^{-1},\ldots,a_{n-1-m}^{-1}),$ and the element $u_n=(1,e,\ldots,e)$ is a strong unit.
Then $K_{I}^{\lambda,\rho}(G)_{ea}$ is isomorphic to the pseudo effect algebra $\Gamma(G_n(G),u_n).$
\end{example}

\begin{example}\label{ex:3.8}
Let $I=\mathbb Z$ and put $\lambda(i)=i$ and $\rho(i)=i-1,$ $i \in I.$ Then the kite pseudo effect algebra $K_\mathbb Z^{\lambda,\rho}(\mathbb Z)_{ea}$  satisfies \RDP$_2.$

Define $W(\mathbb Z):=\mathbb Z\lex \mathbb Z^\mathbb Z,$ and let  multiplication $*$ on it be defined as follows: $(m_1,x_i)*(m_2,y_i)=(m_1+m_2, x_i+y_{i+m_1}).$ Then $(W(\mathbb Z);(0),*)$ is an $\ell$-group, called the wreath product of $\mathbb Z$ by $\mathbb Z$ {\rm \cite[Ex 35.1]{Dar}}, with strong unit $u=(1,(0))$, and the kite pseudo effect algebra $K_\mathbb Z^{\lambda,\rho}(\mathbb Z)_{ea}$ is isomorphic to $\Gamma(W(\mathbb Z),u).$
\end{example}

\begin{example}\label{ex:3.9}
Let $I=\mathbb Z$ and put $\lambda(i)=i$ and $\rho(i)=i-1,$ $i \in I,$ and let $G$ be a po-group.

Define $W(G):=\mathbb Z\lex  G^\mathbb Z,$ and let multiplication $*$ on it be defined as follows: $(m_1,x_i)*(m_2,y_i)=(m_1+m_2, x_iy_{i+m_1}).$ Then $(W(G);(e),*)$ is a po-group, called the wreath product of $G$ by $\mathbb Z,$ with strong unit $u=(1,(e)),$ and the kite pseudo effect algebra $K_\mathbb Z^{\lambda,\rho}(G)_{ea}$ is isomorphic to $\Gamma(W(G),u).$
\end{example}

\begin{example}\label{ex:3.10}
Let $I=\mathbb Z,$ $k\in \mathbb Z$ be fixed, and put $\lambda(i)=i$ and $\rho_k(i)=i-k,$ $i \in I,$ and let $G$ be a po-group.

Define $W(G):=\mathbb Z\lex  G^\mathbb Z,$ and let multiplication $*_k$ on it be defined as follows: $(m_1,x_i)*_k(m_2,y_i)=(m_1+m_2, x_iy_{i+km_1}).$ Then $(W(G);(e),*_k)$ is a po-group with strong unit $u=(1,(e))$ and the kite pseudo effect algebra $K_\mathbb Z^{\lambda,\rho_k}(G)_{ea}$ is isomorphic to $\Gamma(W(G),u).$
\end{example}

\begin{example}\label{ex:3.11}
Let  $\lambda, \rho:I \to I$ be bijections such that $\lambda \circ \rho = \rho\circ \lambda.$ Let $G$ be a po-group and define $W^{\lambda,\rho}_I(G)=\mathbb Z \lex G^I$ and let multiplication $*$ on $W^{\lambda,\rho}_I(G)$ be defined as follows

$$
(m_1,x_i)*(m_2,y_i):=(m_1+m_2, x_{\lambda^{-m_2}(i)}y_{\rho^{-m_1}(i)}).
$$

Then $W^{\lambda,\rho}_I(G)$ is a po-group ordered lexicographically  such that $(0,(e))$ is the neutral element and $(m,x_i)^{-1}=(-m,x^{-1}_{(\rho\circ \lambda)^m(i)}),$ and the element $u=(1,(e))$ is a strong unit. If $G$ is an $\ell$-group, so is $W^{\lambda,\rho}_I(G).$
In addition,  $K_{I}^{\lambda,\rho}(G) \cong \Gamma(W^{\lambda,\rho}_I(G),u).$
\end{example}

The conditions of Example \ref{ex:3.11} are satisfies for example if $\rho= \lambda^k$ for some integer $k \in \mathbb Z.$

\section{Kite Pseudo Effect Algebras and the Riesz Decomposition Properties}

A very important class of pseudo effect algebras is closely tied with the Riesz Decomposition Property RDP$_1.$ In the section, we show how kite pseudo effect algebras are connected with different types of RDP's.

\begin{theorem}\label{th:4.1}
Let $I,$ $\lambda,\rho: I \to I$ be given and let $G$ be a directed po-group. The kite pseudo effect algebra $K^{\lambda,\rho}_I(G)_{ea}$ satisfies \RDP (or \RDP$_1$ or \RDP$_2$, respectively) if and only if $G$ satisfies \RDP (or \RDP$_1$ or \RDP$_2$, respectively).
\end{theorem}

\begin{proof}
Let $K^{\lambda,\rho}_I(G)_{ea}$ satisfy RDP and let, for $a_1,a_2,b_1,b_2 \in G^+,$ we have $a_1a_2=b_1b_2.$ Fix an element $i_0\in I$ and for $i=1,2,$ let us define $A_i=\langle f_j^i \colon j\in I\rangle$ by $f_j^i=a_i$ if $j=i_0$ and $f_j^i=e$ otherwise, $B_i= \langle g_j^i \colon j\in I\rangle$ by $g_j^i=a_i$ if $j=i_0$ and $j_j^i=e$ otherwise. Then $A_1+A_2=B_1+B_2$ so that there are $E_{11}=\langle e^{11}_j \colon j\in I\rangle,$ $E_{12}=\langle e^{12}_j \colon j\in I\rangle,$ $E_{21}=\langle e^{21}_j \colon j\in I\rangle,$ $E_{22}=\langle e^{22}_j \colon j\in I\rangle,$ such that $A_1=E_{11}+E_{12},$ $A_2=E_{21}+E_{22},$ $B_1=E_{11}+E_{21},$ and $B_2=E_{12}+E_{22}.$ Using $(IV)$ of Theorem \ref{th:3.4}, from $j=i_0$ we conclude that the elements $e_{uv}=e^{uv}_{i_0},$ $u,v=1,2,$ form a necessary decomposition for $a_1,a_2,b_1,b_2$ which proves $G$ satisfies RDP.

To prove that $G$ satisfies RDP$_1$ if $K^{\lambda,\rho}_I(G)_{ea}$ satisfies RDP$_1,$  we see that for $j\ne i_0$ we have $e^{12}_j=e=e^{21}_j.$ If now $e\le x\le e_{12}$ and $e\le y\le e_{21}$ Putting $X =\langle x_j \colon j\in I\rangle$ and $Y =\langle y_j \colon j\in I\rangle,$ where $x_j=x,$ $y_j=y$ if $j=i_0$ and $x_j=e$ and $y_j=e$ otherwise, we have $X+Y = Y+X$, consequently, $xy=yx$ which proves $e_{12}\, \mbox{\rm \bf com}\, e_{21},$ and $G$ satisfies RDP$_1.$

In the same way we prove the case with RDP$_2.$

Conversely, suppose $G$ satisfies RDP (or RDP$_1$).

(i) If $\langle f_j\colon j\in I\rangle + \langle g_j\colon j\in I\rangle= \langle h_j\colon j\in I\rangle + \langle k_j\colon j\in I\rangle$ from $(IV)$ of Theorem \ref{th:3.4} we conclude that for them we can find an RDP decomposition or an RDP$_1$ one.

(ii) Assume $\langle a_i^{-1}\colon i\in I\rangle + \langle f_j\colon j\in I\rangle= \langle b_i^{-1}\colon i\in I\rangle +\langle g_j\colon j\in I\rangle.$ Then $a_i,b_i,f_j,g_j \ge e$ for any $i,j \in I.$

Now we will use tables so that we will write the elements of the kite in a simpler way: instead of $\langle a_i^{-1}\colon i\in I\rangle$ and $\langle f_j\colon j\in I\rangle$ we use $\langle a_i^{-1}\rangle$ and $\langle f_j\rangle.$

Since $G$ is directed, for any $i\in I$, there is an element $d_i\in G$ such that $a_i,b_i \le d_i,$ so that $d_i^{-1}\le a_i^{-1},b_i^{-1},$ and $e\le d_ia^{-1}_i,d_ib_i^{-1}.$

From $(II)$ of Theorem \ref{th:3.4}, we have $a_i^{-1} f_{\rho^{-1}(i)}=b_i^{-1} g_{\rho^{-1}(i)}$ for any $i\in I.$ Then $(d_ia_i^{-1}) f_{\rho^{-1}(i)}=(db_i^{-1}) g_{\rho^{-1}(i)},$ hence using RDP for $G$, there are $c_{iuv}\in G^+,$ $u,v=1,2,$ such that we have the decomposition tables:

$$
\begin{matrix}
d_ia_i^{-1}  &\vline & c_{i11} & c_{i12}\\
f_{\rho^{-1}(i)} &\vline & c_{i21} & c_{i22}\\
  \hline     &\vline      &d_ib_i^{-1} & g_{\rho^{-1}(i)}
\end{matrix}\ \
$$
and

$$
\begin{matrix}
a_i^{-1}  &\vline & d^{-1}_ic_{i11} & c_{i12}\\
f_{\rho^{-1}(i)} &\vline & c_{i21} & c_{i22}\\
  \hline     &\vline      &b_i^{-1} & g_{\rho^{-1}(i)}
\end{matrix}\ \ .
$$
Since $e\le d_ia_i^{-1}=c_{i11}c_{i12},$ we have $a_i^{-1}c_{i12}^{-1}=d_i^{-1}c_{i11}\le e.$  This yields that we have an RDP decomposition in the kite for case (ii).

$$
\begin{matrix}
\langle a_i^{-1}\rangle  &\vline & \langle d^{-1}_ic_{i11}\rangle  & \langle c_{\rho(j)12} \rangle\\
\langle f_j \rangle &\vline & \langle c_{\rho(j)21}\rangle & \langle c_{\rho(j)22}\rangle\\
  \hline     &\vline      & \langle b_i^{-1} \rangle & \langle g_j \rangle
\end{matrix}\ \ .
$$

It is evident that if $G$ satisfies RDP$_1,$ the later table gives also an RDP$_1$ decomposition.

(iii) Assume $\langle f_j\colon j\in I\rangle + \langle a_i^{-1}\colon i\in I\rangle= \langle g_j\colon j\in I\rangle +\langle b_i^{-1}\colon i\in I\rangle.$ We follows the ideas from the proof of case (ii). By $(III)$ of Theorem \ref{th:3.4}, we have $f_{\lambda^{-1}(i)}a_i^{-1}= g_{\lambda^{-1}(i)} b_i^{-1}$ so that $f_{\lambda^{-1}(i)}(a_i^{-1}d^{-1}_i)
= g_{\lambda^{-1}(i)} (b_i^{-1}d_i^{-1}),$ where $d_i\ge a_i,b_i.$ The RDP holding in $G$ entails the decompositions

$$
\begin{matrix}
f_{\lambda^{-1}(i)}  &\vline & d_{i11} & d_{i12}\\
a_i^{-1}d_i &\vline & d_{i21} & d_{i22}\\
  \hline     &\vline      &g_{\lambda^{-1}(i)} & b_i^{-1}d_i
\end{matrix}\ \
$$
and

$$
\begin{matrix}
f_{\lambda^{-1}(i)}  &\vline & d_{i11} & d_{i12}\\
a_i^{-1} &\vline & d_{i21} & d_{i22}d_i^{-1}\\
  \hline     &\vline      &g_{\lambda^{-1}(i)} & b_i^{-1}
\end{matrix}\ \ ,
$$
where $d_{i22}d_i^{-1}\le e.$ This implies RDP decomposition for  case (iii)

$$
\begin{matrix}
\langle f_j\rangle  &\vline & \langle d_{\lambda(j)11}\rangle  & \langle d_{\lambda(j)12}\rangle\\
\langle a_i^{-1} \rangle &\vline & \langle d_{\lambda(j)21}\rangle & \langle d_{i22}d_i^{-1}\rangle\\
  \hline     &\vline      & \langle b_i^{-1} \rangle & \langle g_j \rangle
\end{matrix}\ \ .
$$

This decomposition is also an RDP$_1$ decomposition whenever $G$ satisfies RDP$_1.$

(iv) Assume $\langle a_i^{-1}\colon i\in I\rangle + \langle f_j\colon j\in I\rangle= \langle g_j\colon j\in I\rangle +\langle b_i^{-1}\colon i\in I\rangle.$

By $(II)-(III)$ of Theorem \ref{th:3.4}, we have $a_i^{-1}f_{\rho^{-1}(i)} = g_{\lambda^{-1}(i)}b_i^{-1}$ for any $i \in I.$ Therefore, $g^{-1}_{\lambda^{-1}(i)}a^{-1}_i= b_i^{-1}f^{-1}_{\rho^{-1}(i)}\le e.$

We can use the decomposition table

$$
\begin{matrix}
\langle a^{-1}_i\rangle  &\vline & \langle g_j\rangle  & \langle g^{-1}_{\lambda^{-1}(i)}a^{-1}_i\rangle\\
\langle f_j \rangle &\vline & \langle e_j\rangle & \langle f_j\rangle\\
  \hline     &\vline      & \langle g_j \rangle & \langle b_i^{-1} \rangle
\end{matrix}\ \ .
$$

It gives also an RDP$_1$ decomposition table whenever $G$ satisfies RDP$_1.$

Finally, if $G$ satisfies RDP$_2,$ by \cite[Prop 4.2(ii)]{DvVe1}, $G$ is an $\ell$-group and by Theorem \ref{th:3.1}, $K^{\lambda,\rho}_I(G)_{mv}$ is a pseudo MV-algebra, so that it satisfies RDP$_2$. Consequently, $K^{\lambda,\rho}_I(G)_{ea}$ also satisfies RDP$_2.$
\end{proof}

We remark that if the kite pseudo effect algebra of a po-group $G$ satisfies a kind of the Riesz Decomposition Property, then so does $G$ the same property. We do not know whether the converse statement is true without the assumption of directness  of $G.$

\begin{proposition}\label{pr:4.2}
A kite pseudo effect algebra $K^{\lambda,\rho}_I(G)_{ea}$ satisfies \RDP$_0$ if and only if $G$ satisfies \RDP$_0.$
\end{proposition}

\begin{proof}
Using $(IV)$ of Theorem \ref{th:3.4}, it is evident that if $K^{\lambda,\rho}_I(G)_{ea}$ satisfies RDP$_0,$ so does $G.$

Conversely, let $G$ is with RDP$_0.$

(i) Let $\langle f_j\colon j\in I\rangle \le \langle g_j\colon j\in I\rangle + \langle h_j\colon j\in I\rangle= \langle g_jh_j\colon j\in I\rangle$ which implies that for each $j\in I,$ there are $g_{j1},h_{j1} \in G^+$ such that $g_{j1}\le g_j,$  $h_{j1}\le h_j$ and $f_j= g_{j1}h_{j1}.$ Then $\langle f_j\colon j\in I\rangle = \langle g_{j1}\colon j\in I\rangle + \langle h_{j1}\colon j\in I\rangle$ and $\langle g_{j1}\colon j\in I\rangle \le \langle g_{j}\colon j\in I\rangle$ and $\langle h_{j1}\colon j\in I\rangle \le \langle h_{j}\colon j\in I\rangle.$

(ii) Let $\langle f_j\colon j\in I\rangle \le \langle a_i^{-1}\colon i\in I\rangle + \langle g_j\colon j\in I\rangle.$ Then $\langle f_j\colon j\in I\rangle= \langle f_j\colon j\in I\rangle + \langle e_j\colon j\in I\rangle,$ and $\langle f_j\colon j\in I\rangle \le \langle a_i^{-1}\colon i\in I\rangle$ and $\langle e_j\colon j\in I\rangle \le \langle g_j\colon j\in I\rangle.$

In the analogous way we can prove the case $\langle f_j\colon j\in I\rangle \le   \langle g_j\colon j\in I\rangle +\langle a_i^{-1}\colon i\in I\rangle.$

(iii) Let $\langle a_i^{-1}\colon i\in I\rangle \le \langle b_i^{-1}\colon i\in I\rangle + \langle f_j\colon j\in I\rangle.$
For each $i \in I,$ we have $a_i^{-1} \le b_i^{-1}f_{\rho^{-1}(i)}$ and $b_i \le f_{\rho^{-1}(i)}a_i$.
Using RDP$_0$ for $G,$ we can find positive elements $a_{i1}\le a_i$ and $f_{\rho^{-1}(i)1}\le f_{\rho^{-1}(i)}$ such that $b_i=f_{\rho^{-1}(i)1} a_{i1}.$ Then we have $\langle a_i^{-1}\colon i\in I\rangle = \langle a_i^{-1}f^{-1}_{\rho^{-1}(i)1} \colon i\in I\rangle
+ \langle f_{j1} \colon j\in I\rangle,$ and $\langle a_i^{-1}f^{-1}_{\rho^{-1}(i)1} \colon i\in I\rangle \le \langle b_i^{-1}\colon i\in I\rangle$ and $\langle f_{j1} \colon j\in I\rangle \le \langle f_{j} \colon j\in I\rangle.$

In a dual way, we proceed also with the case $\langle a_i^{-1}\colon i\in I\rangle \le   \langle f_j\colon j\in I\rangle + \langle b_i^{-1}\colon i\in I\rangle.$

Summing up all cases, we have that $K^{\lambda,\rho}_I(G)_{ea}$ satisfies RDP$_0.$
\end{proof}

Examples \ref{ex:3.5}--\ref{ex:3.10} describe concrete po-group representations  corresponding to some kite pseudo effect algebras. Therefore, in view of Theorem \ref{th:2.2}, Theorem \ref{th:4.1} and Theorem \ref{th:3.1} it would be also interesting to study the following: given a po-group $G,$ a set $I$ and bijections $\lambda,\rho:I \to I,$ describe the corresponding unital po-group $(G_I^{\lambda,\rho},u)$ such that the kite pseudo effect algebra $K^{\lambda,\rho}_I(G)_{ea}$ of $G$ is isomorphic to $\Gamma(G_I^{\lambda,\rho},u).$ Some partial answers will be done in Theorem \ref{th:5.3} and in Section 6.

\section{Kites and Perfect Pseudo Effect Algebras}

A perfect algebra is characterized roughly speaking by the property that every element is either infinitesimal or co-infinitesimal. This property have also kite pseudo effect algebras, therefore, we study perfect pseudo effect in more detail with their representation.

First we remind some additional notions of theory of pseudo effect algebras. Let $x$ be an element and $n\ge 0$ be an integer. We define $0x:=0,$ $1x:=x,$ and $(n+1)x:= nx +x$ whenever $nx$ and $nx +x$ are defined in $E.$ An element $x\in E$ is said to be {\it infinitesimal} if $nx$ exists in $E$ for any integer $n \ge 1.$ We denote by $\Infinit(E)$ the set of infinitesimal elements of $E.$

If $A,B$ are two subsets of a pseudo effect algebra $E,$ $A+B:=\{a+b: a \in A, b \in B, a+b \in E\},$ and we say that $A+B$ is {\it defined} in $E$ if $a+b$ exists in $E$ for each $a \in A$ and each $b \in B.$ We write $A\leqslant B$ whenever $a\le b$ for all $a\in A$ and all $b \in B.$ In addition, we write $A^-:=\{a^-: a \in A\}$ and $A^\sim :=\{a^\sim : a \in A\}.$

We note that an {\it ideal} of a pseudo effect algebra $E$ is any subset $I$ of $E$ such that (i) if $x,y \in I$ and $x+y$ is defined in $E,$ then $x+y \in I,$ and (ii) $x\le y \in I$ implies $x\in I.$ An ideal $I$ is {\it maximal} if it is a proper subset of $E$ and it is not a proper subset of any ideal $J\ne E.$ An ideal $I$ is {\it normal} if $x+I=I+x$ for any $x \in E,$ where $x+I:=\{x+y: y \in I,\ x+y $ exists in $E\}$ and in the dual way  we define $I+x.$

An analogue of a probability measure is a state. We say that a mapping $s:E \to[0,1]$ is a {\it state} if (i) $s(a+b)=s(a)+s(b)$ whenever $a+b$ is defined in $E,$ and (ii) $s(1).$ A state $s$ is {\it extremal} if from $s=\alpha s_1+(1-\alpha) s_2,$ where $s_1,s_2$ are states on $E$ and $\alpha \in (0,1),$ we conclude $s=s_1=s_2.$ We denote by $\mathcal S(E)$ and $\partial_e \mathcal S(E)$ the set of all states and extremal states, respectively, on $E.$ It can happen that $\mathcal S(E)$ is empty. In general, $\mathcal S(E)$ is either empty, or a singleton or an infinite set. We remind that if $E$ is an effect algebra with RDP, then $E$ has at least one state.

If $s$ is a state on $E,$ then the {\it kernel} of $s,$ i.e. the set $\Ker(s):=\{a\in E: s(a)=0\},$ is a normal ideal of $E.$

We say that a pseudo effect algebra $E$ is {\it perfect} if there are two subsets $E_0,E_1$ of $E$ such that $E_0\cap E_1=\emptyset$ and $E=E_0\cup E_1$ such that
\begin{enumerate}
\item[(a)] $E_i^- =E_i^\sim = E_{1-i},$ $i=0,1,$
\item[(b)] if $x \in E_i,$ $b\in E_j$ and $x+y$ is defined in $E,$ then $i+j\le 1$ and $x+y \in E_{i+j}$ for $i,j=0,1,$
\item[(c)]    $E_0+ E_0$ is defined in $E.$
\end{enumerate}
In such a case, we write $E=(E_0,E_1).$

For example, let $G$ be a po-group, then $E=\Gamma(\mathbb Z \lex G,(1,0))$ is a symmetric perfect pseudo effect algebra. In addition, every kite pseudo effect algebra of $G$ is perfect, indeed, $K^{\lambda,\rho}_I(G)_{ea}=((G^+)^I,(G^-)^I).$

\begin{proposition}\label{pr:5.1}
Let $E=(E_0,E_1)$ be a perfect pseudo effect algebra. Then
\begin{enumerate}
\item[{\rm (i)}] $E_0 \leqslant E_1.$
\item[{\rm (ii)}] $E_0+E_0$ is defined and $E_0+E_0=E_0.$
\item[{\rm (iii)}] $E$ has a unique state $s$, namely $s(E_0)=\{0\}$ and $s(E_1)=\{1\}.$
\item[{\rm (iv)}] $E_0$ is a unique maximal ideal of $E,$ and $E_0$ is a normal ideal.
\item[{\rm (v)}] $E_0=\Infinit(E).$
\item[{\rm (vi)}] If $i+j >1,$ then neither $x+y$ nor $y+x$ are defined in $E.$

\item[{\rm (vii)}] $E=(F_0,F_1),$ then $E_0=F_0$ and $E_1=F_1.$
\end{enumerate}
\end{proposition}

\begin{proof}
(i) We have $E_0^- = E^\sim =E_1$ and $E_1^- = E_1^\sim = E_0.$ Let $x \in E_0$ and $y \in E_1.$  Then $y^\sim \in E_0$ and $y^\sim+x $ is defined in $E,$ so that $x \le y^{\sim-}=y.$

(ii) We have $E_0+E_0$ is defined in $E$ and $E_0+E_0 \subseteq E_0.$ Let now $x \in E_0.$ Then $x= x+0 \in E_0+E_0$ which entails $E_0+E_0=E_0.$

(iii) We define a mapping $s:E \to \{0,1\}$ by $s(x)=i$ whenever $x\in E_i,$ $i=0,1.$ It is clear that $s(1)=1.$ If $a +b$ is defined in $E,$ assume
(a) $a,b \in E_0$ then $=s(a+b)=s(a)+s(b).$ (b) $a\in E_0,$ $b \in E_1,$ then $a+b\in E_1$ so that $1=s(a+b)=s(a)+s(b).$ Similarly if $a\in E_1$ and $b\in E_0,$ then $a+b \in E_1.$ (c) If $a,b \in E_1,$ then $a+b$ is not defined in $E.$

Uniqueness. Let $s_1$ be any state on $E.$ If $a\in E_0,$ then $na$ is defined in $E$ for any integer $n\ge 1.$ Then $s_1(na)=ns_1(a)\le 1,$ so that $s_1(a)\le 1/n$ i.e., $s_1(a)=0.$ Hence, if $b\in E_1$ then $b^-\in E_0$ and $1=s_1(b^-)+s_1(b) = s_1(b),$ which entails $s=s_1.$

(iv) It is clear that $E_0$ is an ideal of $E.$ If $E_0$ is contained in an ideal $I\ne E,$ and if $E_0\ne I,$ there is an element $y \in E_1\cap I.$ Hence, $y^- \in E_0$ and $1=y^-+y \in I$ which is absurd. Hence, $E_0$ is a maximal ideal of $E.$

Normality. By (iii), there is a unique state $s$ on $E$ such that $\Ker(s)=E_0$ which implies $E_0$ is a normal ideal.

(v) By (ii), we have $E_0\subseteq \Infinit(E).$ If $s$ is a unique two-valued state on $E$ guaranteed by (iii), then $a\in \Infinit(E)$ implies $s(a)=0,$ i.e. $a \in \Ker(s)=E_0.$

(vi) It was already proved in the proof of (iii)

(vii) If $E= (F_0,F_1),$ by (iii), $E$ has a unique two-valued state $s$ such that $\Ker(s)=E_0$ and hence, also $\Ker(s)=F_0.$
\end{proof}

\begin{theorem}\label{th:5.2}
Let $E$ be a perfect symmetric pseudo effect algebra with \RDP$_1.$ There is a directed po-group $G$ with \RDP$_1$ such that $E\cong \Gamma(\mathbb Z\lex G,(1,0)).$ The po-group $G$ is unique up to isomorphism.
\end{theorem}

\begin{proof}
In the present proof we will assume that used po-groups are written additively.

Since $E=(E_0,E_1)$ has RDP$_1$, by Theorem \ref{th:2.2}, there is a unique (up to isomorphism of unital po-groups) unital po-group $(H,v)$ such that $E=\Gamma(H,v).$ By Proposition \ref{pr:5.1}(ii), we see that  $(E_0;+,0)$ is a semigroup such that,  for $x,y, z \in E_0,$ (i) $x+y=0$ implies $x=0=y,$ (ii) $x+y=x+z$ implies $y=z$ and $y+x=z+x$ implies $y=z$. Then $(E_{0};+,0)$  is a cancellative semigroup satisfying the conditions of the Birkhoff Theorem, \cite[Thm II.4]{Fuc}, which guarantees that $E_{0}$ is the positive
cone of a unique (up to isomorphism) po-group $G$. Without loss of generality, we can assume that $G$ is generated by the positive cone $E_0,$ so that $G$ is directed. In addition, since $E$ is with RDP$_1$, and $G^+=E_0,$ it is clear that $G$ satisfies RDP$_1.$

Since $E_0$  is a subset of a po-group $H,$ we can assume that $G$ is a subgroup of $H.$

Define a symmetric perfect pseudo effect algebra $\mathbb Z(G):=\Gamma(\mathbb Z \lex G,(1,0))$ and let $\phi$ be a mapping from $E$ into $\mathbb Z(G)$ defined by
$$
\phi(x):=(i,x-i1)
$$
whenever $x \in E_i,$ $i=0,1.$ We note ``$-$" in the formula means the group subtraction taken in the po-group $H.$ It is evident that $\phi$ is a well-defined mapping.

We have $\phi(0)=(0,0),$ $\phi(1)=(1,0).$ If $x \in E_0,$ then $x^-\in E_1$ so that $\phi(x^-)=(1,x^- -1)=(1, 1-x+1)=(1,-x+1-1)=(1,-x)=\phi(x)^-$. If $x\in E_1,$ then $\phi(x^-)=(0,x^-)=(0,-x+1)=\phi(x)^-.$ In a similar way we can prove that $\phi$ is a homomorphism: For example, let $x\in E_0,$ $y\in E_1$ and let $x+y$ be defined in $E$. Then $y\le x^-$ and $\phi(y)= (1,y-1)  \le (1,x^--1) =\phi(x^-)=\phi(x)^-.$ So that $\phi(x+y)=(1,x+y-1)= (0,x)+(1,y-1)= \phi(x)+\phi(y).$

Assume $\phi(x)\le \phi(y)$ for $x \in E_i$ and $y \in E_j$ for some $i,j =0,1.$ Then $\phi(x)=(i,x-i1)\le \phi(y)= (j,y-j1).$ Then either $i<j,$ so that $x<y$ or $i=j$ and then $x-i1\le y-i1$ which gets $x \le y$ and $\phi$ is injective.

Now let $(i,g) \in \mathbb Z(G).$ Then either $(i,g)=(0,g),$ where $g\in G^+=E_0,$ so that $\phi(g)=(0,g)$ or $(i,g)=(1,-h)$ for some $h\in G^+,$ so that $\phi(h^-)=(i,g),$ and $\phi$ is surjective.

Consequently, $\phi$ is an isomorphism in question.
\end{proof}

Now we apply Theorem \ref{th:5.3} to kite pseudo effect algebras.

\begin{theorem}\label{th:5.3}
Let a set $I,$ a bijection $\lambda: I \to I$ and  a directed po-group $G$ with \RDP$_1$ be given. There is a unique (up to isomorphism) directed po-group $G^\lambda_I$ with \RDP$_1$ such that
the kite pseudo effect algebra $K^{\lambda,\lambda}_I(G)_{ea}$ is isomorphic to $\Gamma(\mathbb Z\lex G^\lambda_I,(1,0)).$
\end{theorem}

\begin{proof}
By Theorem \ref{th:3.4}, any kite pseudo effect algebra $K^{\lambda,\rho}_I(G)_{ea}$ is symmetric iff $\lambda =\rho.$ It is evident that every kite pseudo effect algebra is perfect.
Since $G$ is with RDP$_1,$ by Theorem \ref{th:4.1}, the kite $K^{\lambda,\lambda}_I(G)_{ea}$ has also RDP$_1.$ Consequently, all conditions of Theorem \ref{th:5.2} are satisfied, thus there is a unique (up to isomorphism) directed po-group $G^\lambda_I$ with RDP$_1$ such that
$K^{\lambda,\lambda}_I(G)_{ea}\cong \Gamma(\mathbb Z\lex G^\lambda_I,(1,0)).$
\end{proof}

Finally, we would like to note that it would be interesting to prove Theorem \ref{th:5.2} also for perfect pseudo effect algebras which are not necessarily symmetric.

\section{Least Non-trivial Normal Ideals of Kites}

We are studying pseudo effect algebras which are partial algebras, therefore, some notions of universal algebra are not easy to study for pseudo effect algebras. If the pseudo effect algebra is a pseudo MV-algebra, then there is a one-to-one correspondence  between normal ideals and congruences, \cite{GeIo}. For pseudo effect algebras it is not necessarily a case. Therefore, instead of subdirect irreducibility we will study the relation between the least non-trivial (i.e. $\ne \{0\}$) normal ideals of kite pseudo effect algebras and the least non-trivial (i.e. $\ne \{e\}$) o-ideals of the corresponding po-groups. We will be inspired by research from \cite{DvKo}, where subdirect irreducibility of kites connected with  subdirect irreducibility of $\ell$-groups was studied.

We note that a pseudo effect algebra $E$ is a {\it subdirect product} of a system of pseudo effect algebra $(E_t: t \in T)$ if there is an injective homomorphism $h: E\to \prod_{t\in T}E_t$ such that $\pi_t(h(E))=E_t$ for each $t \in T,$ where $\pi_t$ is the projection of $\prod_{t \in T} E_t$ onto $E_t.$ In addition, $E$ is {\it subdirect irreducible} if whenever $E$ is a subdirect product of $(E_t: t \in T),$ there exists $t_0 \in T$ such that $\pi_{t_0}\circ h$ is an isomorphism of pseudo effect algebras.

If $M$ is a pseudo MV-algebra, it is known, \cite{GeIo}, that there is a one-to-one correspondence between congruences and normal ideals. For pseudo effect algebras, this relation is more complicated, and it is true for pseudo effect algebras with RDP$_1$, \cite{DvVe3} or for Riesz ideals of general pseudo effect algebras. For pseudo MV-algebras, we have that a pseudo MV-algebra $M$ is subdirectly irreducible iff the intersection of all non-trivial normal ideals of $M$ is non-trivial, or equivalently, $M$ has the least non-trivial normal ideal.

We remind that by an {\it o-ideal} of a directed po-group $G$ we understand any normal directed convex subgroup $H$ of $G;$ convexity means if $g,h\in H$ and $v\in G$ such that $g\le v\le h,$ then $v \in H.$ If $G$ is a po-group, so is $G/H,$ where $x/H \le y/H$ iff $x\le h_1y$ for some $h_1\in H$ iff $x \le yh_2$ for some $h_2 \in H.$
If $G$ satisfies one of RDP's, then $G/H$ satisfies the same RDP, \cite[Prop 6.1]{174}. Therefore, in this section, we will study relations between the least non-trivial normal ideals of kite pseudo effect algebras and the least non-trivial o-ideals of the corresponding po-groups.

If $E$ satisfies RDP$_1,$ there is a one-to-one correspondence between congruences and ideals of $E,$ \cite{185, DvVe3}, and given a normal ideal $I$, the relation $\sim_I$ on $E$ is given by $a\sim_I b$ iff there are $x,y \in I$ with $x\le a$ and $y\le b$ such that $a\minusli x = b\minusli y,$ then the quotient  $E/I$ is a pseudo effect algebra with RDP$_1.$

We note that if $E=\Gamma(G,u)$ for some unital po-group with RDP$_1,$ there is a one-to-one correspondence between ideals of $E$ and o-ideals in $G,$ \cite{185}.

\begin{proposition}\label{pr:6.1}
Let $E=\Gamma(G,u)$ for some unital po-group with RDP$_1.$ If $I$ is an ideal of $E,$ then the set
$$\phi(I)=\{x \in G: \exists \ x_i,y_j \in I,\ x= x_1+\cdots +x_n-y_1-\cdots -y_m\}$$
is a convex subgroup of $G$ generated by $I.$ The set $\phi(I)$ is an o-ideal if and only if $I$ is a normal ideal, and $E/I =\Gamma(G/\phi(I),u/\phi(I)).$ The mapping $I\mapsto \phi(I)$ is a one-to-one correspondence  between normal ideals of $E$ and o-ideals of $G$ preserving set-theoretical inclusion. If $K$ is an o-ideal, then the restriction $K\cap[0,u]$ is a normal ideal of $E.$
\end{proposition}

In addition, due to \cite[Prop 3.1]{185}, if $a$ is an element of a pseudo effect algebra $E$ with RDP$_0,$ then the normal ideal of $E$ generated by $a$ is the set

\begin{align*}
N_0(a)=\{x\in E: x&=x_1\minusre (a_1+x_1)+\cdots +x_n \minusre (a_n+x_n)\\
&=(y_1+a'_1)\minusli y_1+\cdots +(y_m+a'_m)\minusli y_m,\
 x_i,y_j\in E, a_i,a'_j\le a\},
\end{align*}
assuming that the corresponding elements 
are defined in $E.$

\begin{lemma}\label{le:6.2}
Let $G=(G;+,-,0)$ be a po-group with \RDP$_1$ and $N$ be a normal ideal of the pseudo effect algebra $E=\Gamma(G,u).$ If $x \in N,$ then $x^{--},x^{\sim\sim}\in N.$ More general, if for any $y,z\in E$  such that $x+y$ and $z+x$ are defined in $E,$ the elements $y\minusre (x+y)= -y +x + y$ and $(z+x)\minusli z= z+x-z$ belong to $N,$ where the calculation is taken in the po-group $G.$
\end{lemma}

\begin{proof}
Let $x \in N$ and suppose then $x^{--}=u+x-u$ and $x^{\sim\sim}=-u+x+u,$ where addition and subtraction are defined in the po-group $G.$ Let $\phi(N)$ be the o-ideal of $G$ generated by $N,$ by Proposition \ref{pr:6.1}, the elements $x^{--}$ and $x^{\sim\sim}$ belong to $\phi(N),$ therefore again by Proposition \ref{pr:6.1}, $x^{--},x^{\sim\sim} \in N.$

In the same way we prove that the elements $-y+x+y, z+x-z \in N.$
\end{proof}

We note that every proper ideal of $K^{\lambda,\rho}_I(G)_{ea}$ is a subset of $(G^+)^I.$

An element $\langle f_j\colon j \in I\rangle$ is said to be $\alpha$-{\it dimensional}, for some cardinal $\alpha,$ if $|\{j \in I: f_j \ne e\}|=\alpha.$ In the same  way we define an $\alpha$-dimensional element $\langle a_i^{-1}\colon i \in I\rangle.$ One-dimensional elements are particulary easy to work with. Moreover, every element $\langle f_j\colon j \in I\rangle$ is a join of one-dimensional elements, and any element $\langle a_i^{-1}\colon i \in I\rangle$ is a meet of one-dimensional ones.

\begin{proposition}\label{pr:6.3}
Let $I$ be a set and $\lambda,\rho:I \to I$ be bijections.
If $H$ is an o-ideal of a directed po-group $G$ and $N=H^+,$ then $N^I:=\{\langle f_j\colon j\in I\rangle: f_j \in H,\ j\in I\}$ is a normal ideal of the kite pseudo effect algebra
$K^{\lambda,\rho}_I(G)_{ea}.$ In addition, if $N_f^I$ denote the set of all finite dimensional elements from $N^I,$ then $N_f^I$ is a non-trivial normal ideal of the kite pseudo effect algebra.

Conversely, if $J$ is a normal ideal of $K^{\lambda,\rho}_I(G)_{ea},$ then $\pi_j(J)$ is the positive cone of an o-ideal of $H,$ where
$\pi_j$ is the $j$-th projection of $\langle f_j \colon j \in I\rangle \mapsto f_j.$
\end{proposition}

\begin{proof}
From $(IV)$ of Theorem \ref{th:3.4} we conclude that $N^I$ is an ideal of the kite. To show the normality $x + N^I = N^I+x$ is necessary to assume that $x =\langle a_i^{-1} \colon i \in I\rangle.$ By $(II)-(III)$ of Theorem \ref{th:3.4}, we have $\langle f_j \colon j \in I\rangle + \langle a_i^{-1} \colon i \in I\rangle=
\langle f_{\lambda^{-1}(i)}a_i^{-1} \colon i \in I\rangle.$ If we set $h_j =a_{\rho(j)}f_{\lambda^{-1}(\rho(j))}a_{\rho(j)}^{-1},$
then $h_j \in N^+$ and $\langle a_i^{-1} \colon i \in I\rangle + \langle h_j \colon j \in I\rangle       = \langle f_j \colon j \in I\rangle + \langle a_i^{-1} \colon i \in I\rangle. $ In a similar way we deal with the second case.

Using the same methods as for $N^I,$ we see that $N_f^I$ is a non-trivial normal ideal of the kite pseudo effect algebra.

Now let $J$ be a normal ideal of the kite pseudo effect algebra of $G.$ Then $J\subseteq (G^+)^I$ so that the projection is the positive cone of an o-ideal $H$ of $G.$
\end{proof}

\begin{lemma}\label{le:6.4}
Let $K^{\lambda,\rho}_I(G)_{ea}$ be a kite pseudo effect algebra and $\langle f_j\colon j \in I\rangle$ be one-dimensional such that $f_{j_0} \ne e$ for some $j_0 \in I.$ Then the elements
$\langle f_j\colon j \in I\rangle^{--},$ $\langle f_j\colon j \in I\rangle^{-\sim}=\langle f_j\colon j \in I\rangle,$ $\langle f_j\colon j \in I\rangle^{\sim-}=\langle f_j\colon j \in I\rangle,$ and $\langle f_j\colon j \in I\rangle^{\sim\sim}$ are also one-dimensional. In addition, for $x,y \in K^{\lambda,\rho}_I(G)_{ea},$ $ x \minusre (\langle f_j\colon j \in I\rangle + x)$ and $y \minusli (y + \langle f_j\colon j \in I\rangle)$ are one-dimensional assuming that the corresponding elements exist.

\begin{enumerate}
\item[{\rm (1)}] $\langle g_j\colon j \in I\rangle \minusre (\langle f_j\colon j \in I\rangle + \langle g_j\colon j \in I\rangle)= \langle g_j^{-1}f_jg_j\colon j \in I\rangle,$

$(\langle g_j\colon j \in I\rangle + \langle f_j\colon j \in I\rangle)\minusli \langle g_j\colon j \in I\rangle= \langle g_jf_jg_j^{-1}\colon j \in I\rangle.$

\item[{\rm (2)}] $ (\langle f_j\colon j \in I\rangle + \langle a_i^{-1}\colon i \in I\rangle)\minusli \langle f_j\colon j \in I\rangle =\langle f_{\lambda^{-1}(i)} a_i^{-1} f^{-1}_{\rho^{-1}(i)}\colon i \in I\rangle,$ if $f_{\lambda^{-1}(i)} \le a_i,$

$\langle f_j \colon j \in I\rangle \minusre (\langle a_i^{-1}\colon i \in I\rangle + \langle f_j\colon j \in I\rangle)= \langle f_{\lambda^{-1}(i)}^{-1} a_i^{-1} f_{\rho^{-1}(i)}\colon i \in I\rangle,$ if $f_{\rho^{-1}(i)} \le a_i,$

\item[{\rm (3)}] $\langle f_j\colon j \in I\rangle^{--}\wedge \langle f_j\colon j \in I\rangle=0$ iff $\rho(\lambda^{-1}(j))\ne j.$

\item[{\rm (4)}] $\langle f_j\colon j \in I\rangle^{\sim\sim}\wedge \langle f_j\colon j \in I\rangle=0$ iff $\lambda(\rho^{-1}(j))\ne j.$
\end{enumerate}
\end{lemma}

\begin{proof}
The proofs follow easily from the definition of the right and left negations presented in Proposition \ref{pr:3.2} and from the fact that both $\rho\circ\lambda^{-1}$ and $\lambda\circ \rho^{-1}$  are bijections.
\end{proof}

\begin{proposition}\label{pr:6.5}
Let a kite pseudo effect algebra of $G,$ $K^{\lambda,\rho}_I(G)_{ea},$ have the least non-trivial normal ideal. Then $G$ has the least non-trivial o-ideal.
\end{proposition}

\begin{proof}
Suppose the converse, i.e. $G$ has no least non-trivial o-ideal. Then there exists a set $\{H_t: t \in T\}$ of non-trivial o-ideals of $G$ such that $\bigcap_{t \in T} H_t =\{e\}.$ Set $N_t=H_t^+,$ $t \in T.$ By Proposition \ref{pr:6.3}, every $N_t^I$ is a normal ideal of the kite. Hence, $\bigcap_{t \in T} H_t \ne \{0\}$ and there is a non-zero element $f=\langle f_j\colon j \in I\rangle\in \bigcap_{t\in T} N_t^I.$ Then, for every index $j \in I,$ $f_j \in H_t$ for every $t \in T$ which yields $f_j=e$ for each $j \in I$ and $f=\langle f_j\colon j \in I\rangle = 0,$ which is a contradiction. Therefore, $G$ has the least non-trivial o-ideal.
\end{proof}

\begin{theorem}\label{th:6.6}
Let $I$ be a set and $\lambda,\rho:I \to I$ be bijections and let $G$ be a directed po-group with \RDP$_1.$ Let $K^{\lambda,\rho}_I(G)_{ea}$ be a kite pseudo effect algebra of a po-group $G.$ The following are equivalent:

\begin{enumerate}
\item[{\rm (1)}] $G$ has the least non-trivial o-ideal and for all $i,j\in I$ there exists an integer $m\ge 0$ such that $(\rho\circ\lambda^{-1})^m(i)=j$ or $(\lambda\circ \rho^{-1})^m(i)=j.$
\item[{\rm (2)}] $K^{\lambda,\rho}_I(G)_{ea}$ has the least non-trivial normal ideal.
\end{enumerate}
\end{theorem}

\begin{proof}
(1) $\Rightarrow$ (2). By Theorem \ref{th:4.1}, the pseudo kite effect algebra $K^{\lambda,\rho}_I(G)_{ea}$ satisfies RDP$_1.$

Let $H$ be the least non-trivial o-ideal of $G$ and let $N=H^+.$
By Proposition \ref{pr:6.3}, $N^I$ and $N_f^I$ are normal ideals of the kite $K^{\lambda,\rho}_I(G)_{ea}.$ Therefore, it is necessary to show that $N^I_f$ is the least non-trivial normal ideal of the kite. In other words, we have to prove that the normal ideal of the kite generated by any nonzero element $\langle f_j\colon j \in I\rangle\in N^I_f$ equals to $N^I_f.$ It is equivalent to show the same for any one-dimensional element from $N^I_f.$ Indeed,
let $f=\langle f_j\colon j \in I\rangle$ be any element from $N^I\setminus \{0\}.$ There is a one-dimensional element $g=\langle g_j\colon j \in I\rangle\in N^I$ such that $0<\langle g_j\colon j \in I\rangle \le \langle f_j\colon j \in I\rangle.$ Hence, $N^I_f=N_0(g)\subseteq N_0(f)\subseteq N^I_f.$

Thus let $g$ be any one-dimensional element from $N^I_f$ and
let $N_0(g)$ be the normal ideal of the kite generated by the element $g.$
Without loss of generality assume $g=\langle g_0,e,\ldots\rangle,$ where $g_0>e,$ $g_0 \in G$; this is always possible by a suitable reordering of $I,$ regardless of its generality. Then $g_0$ generates $N$, that is $N=\{x\in G^+: x= x_1^{-1}g_1x_1\cdots x_n^{-1}g_nx_n=y_1g'_1y_1^{-1}\cdots y_mg'_m y_m^{-1},\ e\le g_i,g'_j\le g_0,\  x_i,y_j\in G^+,\ i=1,\ldots,n,\ j=1,\ldots,m,\  n,m\ge 1\},$ since $H$ is the least non-trivial o-ideal of $G.$ Doing double negations $m$ times of $\langle f_j\colon j \in I\rangle$, we obtain that  either $ \langle f_j\colon j \in I\rangle^{--m}=\langle f_{(\rho\circ \lambda^{-1})^m(j)}\colon j \in I\rangle$ and it belongs to $N_0(g)$ or $\langle f_j\colon j \in I\rangle^{\sim\sim m}=\langle f_{(\lambda\circ\rho^{-1})^m(j)}\colon j \in I\rangle$ which also belongs to $N_0(g).$  Consequently, for any $j \in I,$ there is an integer $m$ such that $(\rho\circ\lambda^{-1})^m(j)=0$ or $(\lambda\circ \rho^{-1})^m(j)=0,$ so that the one-dimensional element whose $j$-th coordinate is $g_0$ is defined in $N_0(g)$ for any $j \in J$; it is either $g^{--m}$ or $g^{\sim\sim m}.$ From Lemma \ref{le:6.4}(i), we see that $f^{-1}g_0f$ and $kg_0k^{-1}$ belong to $N_0(g_0)$ for all $g,k\in G^+,$ which yields, for every $g \in N,$ the one-dimensional element $\langle g,e,\ldots \rangle$ belongs to $N_0(g_0),$ and finally, every one-dimensional element $\langle \ldots, g,\ldots \rangle$  from $N^I_f$ belongs also to $N_0(g_0).$

Now let $J_0=\{j_1,\ldots,j_n\}$ be an arbitrary finite subset of $J,$ $|J_0|=m\ge 1,$ and choose arbitrary $m$ elements $h_{j_1},\ldots,h_{j_m} \in N.$ Define an $m$-dimensional element $g_{J_0}=\langle g_j\colon j \in I\rangle,$ where $g_j=h_{j_k}$ if $j=j_k$ for some $k=1,\ldots,m,$ and $g_j=e$ otherwise. In addition, for $k=1,\ldots,m,$ let $g_k=\langle f_j\colon j \in I\rangle,$ where $f_j=h_{j_k}$ if $j=j_k$ and $f_j=e$ if $j\ne j_k.$ Then $g_{J_0}=g_1+\cdots+g_k\in N_0(g_0).$

Consequently, $N_0(g_0)=N_f^I.$

(2) $\Rightarrow$ (1). By Proposition \ref{pr:6.5}, $G$ has the least non-trivial o-ideal, say $H_0$ and let $N_0=H_0^+.$ Suppose that (1) does not hold. Then for all different $i,j\in I$ and every integer $m\ge 0,$ we have $(\rho\circ\lambda^{-1})^m(i)\ne j$ and $(\lambda\circ \rho^{-1})^m(i)\ne j.$ According to \cite[Thm 5.5]{DvKo}, such indices $i$ and $j$ are said to be disconnected; otherwise they are called connected. If all distinct elements of a subset $K$ of $I$ are connected, $K$ is said to be a connected component of $I.$ By the assumption, there are two elements $i_0,j_0\in I$ which are disconnected. Then for example, the elements $\rho\circ\lambda^{-1}(i_0)$ and $i_0$ are connected.  Let $I_0$ and $I_1$ be a maximal set of mutually connected elements containing $i_0$ and $j_0,$ respectively. Then no element of $I_0$ is connected to any element of $I_1.$

We define $N_0^{I_0}$ as the set of all elements $\langle f_j\colon j \in I\rangle$ such that $j\notin I_0$ implies $f_j=e.$ In a similar way we define $N^{I_1}.$ Then both sets are non-trivial normal ideals of the kite. For example,  as in the proof of Proposition \ref{pr:6.3}, let $\langle f_j \colon j \in I\rangle\in N_0^{I_0}.$ Then $\langle f_j \colon j \in I\rangle + \langle a_i^{-1} \colon i \in I\rangle=
\langle f_{\lambda^{-1}(i)}a_i^{-1} \colon i \in I\rangle.$ If we set $h_j =a_{\rho(j)}f_{\lambda^{-1}(\rho(j))}a_{\rho(j)}^{-1},$
then $\langle h_j \colon j \in I\rangle \in N_0^{I_0}$ and $\langle a_i^{-1} \colon i \in I\rangle + \langle h_j \colon j \in I\rangle = \langle f_j \colon j \in I\rangle + \langle a_i^{-1} \colon i \in I\rangle.$ Similarly for the other possibilities. In the same way we can show that $N_0^{I_1}$ is a non-trivial normal ideal of the kite.

On the other hand, we have $N_0^{I_0} \cap N_0^{I_1} = \{e\}$ which contradicts that the kite has the least non-trivial normal ideal.
\end{proof}

Theorem \ref{th:6.6} has important consequences.

\begin{theorem}\label{th:6.7}
Let $|I|=n$ for some $n\ge 0,$ $\lambda,\rho: I \to I$ be bijections and $G$ be a directed po-group with \RDP$_1.$ If the kite pseudo effect algebra $K^{\lambda,\rho}_I(G)_{ea}$ has the least non-trivial normal ideal, then $G$ has the least non-trivial o-ideal and $K^{\lambda,\rho}_I(G)_{ea}$ is isomorphic to one of:
\begin{enumerate}
\item[{\rm (1)}] $K^{\emptyset,\emptyset}_0(G),$ if $n=0,$ $K^{id,id}_{1}(G)$ if $n=1.$

\item[{\rm (2)}] $K^{\lambda,\rho}_n(G)_{ea}$ for $n\ge 1$ and $\lambda (i)=i$ and $\rho(i) = i-1\ (\mathrm{mod}\, n).$

\end{enumerate}
\end{theorem}

\begin{proof}
In any case, $G$ has the least non-trivial o-ideal.

If $I$ is empty, the only bijection from $I$ to $I$ is the empty function. The kite pseudo effect algebra $K^{\emptyset,\emptyset}_0(G)$ is a two-element Boolean algebra. If $n=1,$ the kite $K^{\emptyset,\emptyset}_1(G)$ is isomorphic to the perfect symmetric pseudo effect algebra $\Gamma(\mathbb Z\lex G,(1,0)).$

For $n\ge 1,$ we can assume without loss of generality that $\lambda$ is the identity map and $\rho$ is a permutation of $I=\{0,1,\ldots,n-1\}.$ If $\rho$ is not cyclic, then there are $i,j \in I$ such that $j$ does not belongs to the orbit of $i,$ which means that $i$ and $j$ are not connected which contradicts Theorem \ref{th:6.6}. So $\rho$ must be cyclic. We can then renumber $I=\{0,1,\ldots,n-1\}$ following the $\rho$-cycle, so that $\rho(i) = i-1\ (\mathrm{mod}\, n),$ $i =0,1,\ldots,n-1.$
\end{proof}

Now we study kites $K^{\lambda,\rho}_I(G)_{ea}$ such that $I$ is infinite. We show that if the kite has the least non-zero normal ideal, then $I$ is at most countable.

\begin{lemma}\label{le:6.8}
Let $I$ be a set and $\lambda,\rho:I \to I$ be bijections and let $G$ be a directed po-group with \RDP$_1.$ If the kite pseudo effect algebra $K^{\lambda,\rho}_I(G)_{ea}$ has the least non-trivial normal ideal, then $I$ is at most countable.
\end{lemma}

\begin{proof}
Suppose $I$ is uncountable and choose an element $i\in I.$ Consider the set $P(i)=\{(\rho\circ\lambda^{-1})^m(i):\ m \ge 0\} \cup \{(\lambda\circ \rho^{-1})^m(i): m\ge 0\}.$ The set $P(i)$ is at most countable, so there is $j \in I\setminus P(i).$ But $P(i)$ exhausts all finite paths alternating $\lambda$ and $\rho$ starting from $i.$ Then $i$ and $j$ are disconnected which contradicts Theorem \ref{th:6.6}. Hence, $I$ is at most countable.
\end{proof}

For example, the kite pseudo effect algebra $K_\mathbb Z^{\lambda,\rho}(\mathbb Z)_{ea},$ from Example \ref{ex:3.8}, where $\lambda(i)=i$ and $\rho(i)=i-1,$ $i\in I=\mathbb Z,$ satisfies the conditions of Theorem \ref{th:6.6}, so that it has the least non-trivial normal ideal; it consists of $\langle f_j \colon j \in I\rangle$ such that $f_j = e$ for all but finite number of indices $j\in I.$

Finally, we present kite pseudo effect algebras having the least non-trivial normal ideal when $I$ is infinite. In such a case by Lemma \ref{le:6.8}, $I$ has to be countable.

\begin{theorem}\label{th:6.9}
Let $|I|=\aleph_0,$ $\lambda,\rho: I \to I$ be bijections and $G$ be a directed po-group with \RDP$_1.$ If the kite pseudo effect algebra $K^{\lambda,\rho}_I(G)_{ea}$ has the least non-trivial normal ideal, then $K^{\lambda,\rho}_I(G)_{ea}$ is isomorphic to $K^{id,\rho}_\mathbb Z(G)_{ea},$ where $\rho(i)=i-1,$ $i \in \mathbb Z.$
\end{theorem}

\begin{proof}
After reordering, without loss of generality, we can assume that $\lambda$ is the identity on $I.$ If $\rho$ is not cyclic, there would be two elements which should be disconnected which is impossible. Therefore, there is an element $i_0\in I$ such that the orbit $P(i_0):=\{\rho^m(i_0): m \in \mathbb Z\}=I.$ Hence, we can assume that $I=\mathbb Z,$ and $\rho(i)=i-1,$ $i \in \mathbb Z.$ Indeed, if we set $j_m=\rho^{-m}(i_0),$ $m\in \mathbb Z,$ then $\rho(j_m)=j_{m-1}$ and we have $\rho(i)=i-1,$ $i \in \mathbb Z.$
\end{proof}

\section{Conclusion}

In the paper we have extended the supply of interesting examples of pseudo effect algebras which are closely connected with po-groups. These pseudo algebras, called kite pseudo effect algebras, are ordinal sum of the product of the positive and the product of the negative cone of a given po-group, where operations depend on two bijections of an index set.  The algebras can be both commutative and non-commutative, and starting with an Abelian po-group it can happen that the corresponding kite pseudo effect algebra can be non-commutative. If we start with an $\ell$-group, the corresponding kite is always a pseudo MV-algebra which model non-commutative many-valued reasoning. The resulting algebra depends on the used bijections.

We have showed how kite pseudo effect algebras are connected with different types of the Riesz Decomposition Properties, Theorem \ref{th:4.1}. We have studied kite pseudo effect algebras also as perfect pseudo effect algebras, giving their representation, Theorem \ref{th:5.3}. Since pseudo effect algebras are partial algebras, it is not straightforward how to study some notions of universal algebras. Therefore, instead of subdirect irreducibility, we study the property when a kite pseudo effect algebra has the least non-trivial normal ideal, Section 6. It was shown that this property is closely connected with the existence of the least non-trivial o-ideal and with special behavior of the used bijections. Some representation theorems, Theorem \ref{th:6.7} and Theorem \ref{th:6.9}, are presented.
Finally some questions are formulated.

The paper has opened a new door in the realm of pseudo effect algebras and it has again showed importance of po-groups and $\ell$-groups for theory of quantum structures which can be useful also for modeling events of quantum measurements when commutativity is not a priori guaranteed.

\end{document}